\documentclass[a4paper, 12 pt, twoside, reqno]{amsart}

\usepackage{tikz}

\usetikzlibrary{arrows.meta,positioning}

\tikzset{
  >=Stealth,
  arr/.style={->, line width=1.1pt},
  lab/.style={midway, fill=white, inner sep=1pt},
  bluelab/.style={lab, text=blue!70!black},
  op/.style={text=black}
}

\newcommand{\Star}[4]{%
\begin{tikzpicture}[baseline={(current bounding box.center)}]
  \coordinate (C) at (0,0);
  \coordinate (L) at (-1.2,0);
  \coordinate (R) at ( 1.2,0);
  \coordinate (U) at (0, 1.2);
  \coordinate (D) at (0,-1.2);

  \draw[arr] (L) -- node[bluelab,above] {$#1$} (C);
  \draw[arr] (C) -- node[bluelab,below] {$#4$} (R);

  \draw[arr] (C) -- node[bluelab,right] {$#2$} (U);
  \draw[arr] (D) -- node[bluelab,left] {$#3$} (C);
\end{tikzpicture}%
}

\newcommand{\Plaquette}[4]{%
\begin{tikzpicture}[baseline={(current bounding box.center)}]
  \coordinate (TL) at (0,1.1);
  \coordinate (TR) at (1.1,1.1);
  \coordinate (BR) at (1.1,0);
  \coordinate (BL) at (0,0);

  \draw[arr] (TL) -- node[bluelab,above] {$#1$} (TR);
  \draw[arr] (BR) -- node[bluelab,right] {$#2$} (TR);
  \draw[arr] (BL) -- node[bluelab,below] {$#3$} (BR);
  \draw[arr] (BL) -- node[bluelab,left] {$#4$} (TL);
\end{tikzpicture}%
}

\usepackage[utf8]{inputenc}
\usepackage{slashed}
\usepackage[english]{babel}
\usepackage[all]{xy}
\usepackage{enumerate}
\usepackage{setspace}
\usepackage{hyperref}

\usepackage[euler-digits,euler-hat-accent]{eulervm}
\usepackage{amsthm}

\usepackage{times}
\usepackage{amsmath}
\usepackage{amsfonts}
\usepackage{amssymb}
\usepackage{amsthm}
\usepackage{graphicx}
\usepackage{array}
\usepackage{color}
\usepackage{mathrsfs}

\usepackage{eucal}
\usepackage{esint} %%% for \fint
\usepackage{tikz}
\usepackage{upgreek}

\usepackage{mathtools}

\allowdisplaybreaks

\theoremstyle{plain}
\newtheorem{theorem}{Theorem}[section]

\newtheorem{corollary}[theorem]{Corollary}
\newtheorem{proposition}[theorem]{Proposition}
\newtheorem{lemma}[theorem]{Lemma}

\theoremstyle{definition}
\newtheorem{definition}[theorem]{Definition}

\theoremstyle{remark}
\newtheorem{remark}[theorem]{Remark}

\numberwithin{equation}{section}
\numberwithin{figure}{section}
\numberwithin{table}{section}

\newcommand{\SM}{{\mathtt s}}
\newcommand{\Aa}{\mathcal{A}}
\newcommand{\A}{\mathcal{A}}

\newcommand{\Cc}{\mathcal{C}}

\newcommand{\R}{\mathbb{R}}
\newcommand{\N}{\mathbb{N}}
\newcommand{\C}{\mathbb{C}} 

\newcommand{\Z}{\mathbb{Z}}
\newcommand{\T}{\mathbb{T}}

\newcommand{\G}{\mathcal{G}}
                 
\newcommand{ \ii}{\,\mathrm{i}\,}
\newcommand{\Ll}{\mathcal{L}}
\newcommand{\Bb}{\mathcal{B}}
                \newcommand{\ie}{\textsl{i.\,e.\,}}

\newcommand{\Kk}{\mathcal{K}}

\DeclareMathOperator{\supp}{supp}
\DeclareMathOperator{\sign}{sign}
\newcommand{\Hh}{\mathcal{H}}

\hypersetup{
pdftoolbar=true,        % show Acrobat’s toolbar?
pdfmenubar=true,        % show Acrobat’s menu?
pdffitwindow=true,     % window fit to page when opened
pdfstartview=true,    % fits the width of the page to the window
pdftitle={Quantum doubles},    % title
pdfauthor={Danilo Polo Ojito},     % author
breaklinks=true, %permet le retour a la ligne dans les liens trop longs
colorlinks=true,       % false: boxed links; true: colored links
linkcolor=purple,         % color of internal links (change box color 
citecolor=teal, % color of links to bibliography
urlcolor=blue, %couleur des hyperliens
bookmarksopen=true, %si les signets Acrobat sont crees,
filecolor=magenta,      % color of file links
}

\begin{document}

\title{The Full Set of KMS-States for Abelian Kitaev Models}

% \author[N. Higson]{Nigel Higson}
% \address{Department of Mathematics, Pennsylvania State University\\ University Park, PA 16802, USA
% \\
% \href{mailto:higson@psu.edu}{higson@psu.edu}}

\author[D. P. Ojito]{Danilo Polo Ojito}
\address{Department of Physics, Universidad de los Andes
	\\Bogota, Colombia \\
	\href{mailto:d.poloo@uniandes.edu.co}{d.poloo@uniandes.edu.co}}

\author[E. Prodan]{Emil Prodan}

\address{Department of Physics and
 Department of Mathematical Sciences 
\\Yeshiva University, New York, NY 10016, USA \href{mailto:prodan@yu.edu}{prodan@yu.edu}
}

\vspace{2mm}

\date{\today}

\maketitle

\begin{abstract}
We first prove that the subalgebra $\Cc$ generated by the vertex and face operators of an abelian Kitaev model is a $C^\ast$-diagonal of the UHF algebra $\Aa$ of quasilocal observables. This gives us access to the Weyl groupoid $\G_\Cc$ associated with the $C^\ast$-inclusion $\Cc \hookrightarrow \Aa$, which supplies a valuable presentation of $\Aa$ as a groupoid $C^\ast$-algebra where the dynamics of the model are generated by a groupoid 1-cocycle $c_H$. Making appeal to the notion of $(c_H,\beta)$-KMS measures for this groupoid, we identify the full set of KMS states of the model and prove its uniqueness for $\beta \in [0,\infty)$. Furthermore, we show that its limit at $\beta \rightarrow \infty$ exists and coincides with the unique frustration-free ground state of the model.


\end{abstract}
\tableofcontents
\section{Introduction}
In his influential paper, A.~Kitaev introduced the \emph{quantum double models}, a class of topologically ordered quantum spin systems defined on triangulations of two-dimensional surfaces \cite{Kit1}. The input data for these models is just a finite group $G$, from which Kitaev constructs a local Hamiltonian consisting of commuting projections and whose ground states exhibit long-range entanglement and support (quasi-)particle excitations with braid (anyonic) statistics described by the representation category of the quantum double algebra $\mathscr{D}(G)$ \cite{BHNV1,Kit2,Wang,Naaij1,Naaij2}. Another important feature is its topological nature: the ground space degeneracy depends on the topology of the underlying surface \cite{Kit1}.

\medskip 
 
When the surface coincides with the $2$-dimensional plane, significant progress has been made in the understanding of equilibrium states of these models at low temperatures.  In this case, the ground state which minimizes the energy locally is non-degenerate, and consequently there exists a unique \emph{frustration-free} ground state separated by a spectral gap from the rest of the spectrum \cite{Kit1,Naaij1,Penneys}. In the infinite-volume limit, however, additional ground states may arise for which the frustration-free condition no longer holds.  The abelian case is completely understood: the manifold of infinite-volume ground states decomposes into $|G|^2$ distinct sectors corresponding to the different types of abelian anyons (i.e., superselection sectors) \cite{CNN}. More recently, this result was extended in the PhD thesis \cite{Hamdan2024} to the non-abelian setting, where a family of ground states labeled by representations of the quantum double $\mathscr{D}(G)$ was constructed. It remains an open problem whether this family exhausts the full set of ground states in the non-abelian case. 

At positive temperatures, much less is known about the equilibrium states, \ie the set of KMS states. In \cite{AlickiJPA2007}, the authors analyze the simplest case $G=\mathbb{Z}_2$ and provide strong evidence for the existence of a unique KMS state at any inverse temperature $\beta$. Their formal argument relies on a reduction to the commutative sub-$C^\ast$-algebra $\Cc$ generated by star and face operators of the toric code, which reveals a connection with the free Ising model. The known results for the latter model imply the uniqueness of the KMS states for the toric code at finite $\beta$. To the best of our knowledge, nothing is known beyond the $G=\mathbb Z_2$ case.

Reference \cite{AlickiJPA2007} made us aware of the relevance of $C^\ast$-inclusions, specifically of the embedding of $\Cc$ into the $C^\ast$-algebra $\Aa$ of quasilocal observables, for the analysis of Kitaev models. In \cite{DE1}, we pointed out that there is a well-developed framework \cite{Kum,Ren1} to investigate such $C^\ast$-inclusions: If $\Cc$ turns out to be a $C^\ast$-diagonal of $\Aa$, then the latter accepts a presentation as the $C^\ast$-algebra of a groupoid whose unit space coincides with the Gelfand spectrum of $\Cc$. In the present work, we show that this is indeed the case for general Abelian Kitaev models. Furthermore, we use a result by Komura \cite{KomuraDocMath2025} together with the fact that the dynamics generated by an Abelian Kitaev model leaves $\Cc$ pointwise invariant, to show that the dynamics is induced from a groupoid 1-cocycle. In turn, this finding enables us to place the arguments of \cite{AlickiJPA2007} in a proper and more general context. Indeed, the characteristics of the mentioned groupoid and a result by J. Renault \cite{RenaultBook} ensure that any KMS state on $\Aa$ is induced from a KMS measure on $\Cc$. We devise an algorithm to compute these measures and to ultimately prove the existence and uniqueness of the KMS states for general finite Abelian groups $G$  at every inverse temperature $\beta \in [0,\infty)$. An important implication of our result is that, in the zero-temperature limit $\beta\to \infty$, the family of KMS states $(\SM_\beta)_{\beta\in[0,\infty)}$ converges to the unique frustration-free ground state.

\medskip 

We finally note that any abelian subalgebra $\Bb \subset \Aa$ can be embedded into a maximal abelian subalgebra $\tilde{\Bb} \subset \Aa$ by an application of Zorn’s lemma. This suggests that, in principle, the techniques developed in this work may extend to any model generated by commuting projections. However, in practice, identifying such a maximal abelian subalgebra is non-trivial, as it requires constructing a maximal family of commuting observables. A particularly interesting direction, currently under investigation, is the case of non-abelian Kitaev models, where the commutation relations are more intricate.

\medskip

\noindent{\bf Acknowledgements:} This work was supported by the U.S. National Science Foundation through the grant CMMI-2131760, and by U.S. Army Research Office through contract W911NF-23-1-0127. The authors acknowledge fruitful discussions with Nigel Higson and Jaime Gomez.

\section{$C^\ast$-Diagonals Generated by Abelian Kitaev Models}

The first part of this section introduces the geometric and algebraic fabric of Abelian Kitaev models, as well as the basic operators, such as the vertex and face local operators and the Hamiltonian itself. The second part of the section is focused on the commutative sub-$C^\ast$-algebra generated by the vertex and face operators, whose space of pure states is analyzed in detail. The last part of the section is dedicated to the proof that this commutative sub-$C^\ast$-algebra is a $C^\ast$-diagonal of the algebra of quasilocal observables.

\subsection{Abelian Kitaev models} We consider abelian Kitaev models \cite{Kit1} over the square lattice $\Ll=\Z^2$, whose sets of edges $e$ and of vertices $v$ are denoted by $E$ and $V$, respectively, and endowed with the discrete topology. We write $e \ni v$ to indicate that edge $e$ emanates from the vertex $v$. A natural orientation is fixed on $E$ by choosing all vertical edges to point upwards, and the horizontal ones pointing to the right, see figure \ref{fig: lattice}. An equally important role is played by the dual graph $\tilde \Ll$, which is equipped with the natural dual orientation. The faces of $\Ll$ will be oriented counterclockwise and will be identified with the vertices of $\tilde \Ll$. The notation $e\in \tilde v$ will specify that edge $e\in E$ belongs to the boundary of face $\tilde v$. Moreover, $\Ll$ and $\tilde \Ll$ will be drawn in the same plane, to make sense of statements like $\rho \cap \tilde \rho \neq \emptyset$ for paths of the direct and dual lattices. We denote by $\Kk(X)$ the family of compact subsets of a topological space $X$, and by $1_S$ the indicator function of a subset $S\subseteq X$.

\begin{figure}
    \centering
    \begin{tikzpicture}[scale=1.2]

\def\N{3}
\def\eps{0.5}

% define N-1, N-2 evaluados
\pgfmathtruncatemacro{\NmOne}{\N-1}
\pgfmathtruncatemacro{\NmTwo}{\N-2}

% vertices primal
\foreach \x in {0,...,\N} {
  \foreach \y in {0,...,\N} {
    \fill (\x,\y) circle (1.2pt);
  }
}

% horizontal primal (right)
\foreach \x in {0,...,\NmOne} {
  \foreach \y in {0,...,\N} {
    \draw[->, thick] (\x,\y) -- (\x+1,\y);
  }
}

% vertical primal (up)
\foreach \x in {0,...,\N} {
  \foreach \y in {0,...,\NmOne} {
    \draw[->, thick] (\x,\y) -- (\x,\y+1);
  }
}

% vertices dual
\foreach \x in {-1,0,...,\NmOne} {
  \foreach \y in {-1,0,...,\NmOne} {
    \fill[red] (\x+\eps,\y+\eps) circle (1.2pt);
  }
}

% horizontal dual (left)
\foreach \x in {-0,...,\NmOne} {
  \foreach \y in {-1,...,\NmOne} {
    \draw[->, red, thick] (\x+\eps,\y+\eps) -- (\x-1+\eps,\y+\eps);
  }
}

% vertical dual (up)
\foreach \x in {-1,...,\NmOne} {
  \foreach \y in {-1,...,\NmTwo} {
    \draw[->, red, thick] (\x+\eps,\y+\eps) -- (\x+\eps,\y+1+\eps);
  }
}
\end{tikzpicture}
    \caption{Black arrows represent the lattice  $\Ll=\mathbb{Z}^2$ with its chosen orientation, while the dual lattice $\tilde{\Ll}$ is depicted  by red arrows.}
    \label{fig: lattice}

\end{figure}
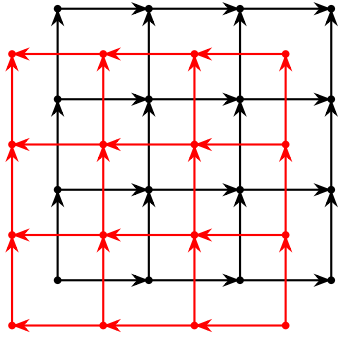
 
 Let $G$ be a non-trivial finite abelian group and denote by $\widetilde{G}$ its Poincaré dual. Each edge $e \in E$ carries the finite-dimensional Hilbert space $\Hh_e=\ell^2(G)$ and the matrix algebra $\mathcal{A}_e=\mathbb{B}(\ell^2(G))\simeq {\rm M}_{n}(\C)$, $n=|G|$. For any $\Lambda \in \Kk(E)$, we may define the finite Hilbert space $\Hh_\Lambda:=\bigotimes_{e\in \Lambda}\Hh_e$ and the $C^\ast$-algebra $\A_\Lambda=\mathbb{B}(\Hh_\Lambda)$ of linear operators over $\Hh_\Lambda$. The local algebra of observables is given by
$\A_{\rm loc}:=\bigcup_{\Lambda\in \Kk(E)} \A_{\Lambda}$
with natural inclusions $a\mapsto a\otimes {\bf1}_{\Lambda\setminus \Lambda'}$ for $\Lambda'\subset \Lambda$. The $C^\ast$-closure of this algebra is called the \emph{quasilocal algebra} of observables, and we shall denote it by $\A$. Obviously, it is isomorphic to the UHF algebra ${\rm M}_{n^\infty}$.

For any element $g\in G$, we denote its inverse by $\bar{g}$. Consider the unitary operators acting on $\ell^2(G)$ by
$$T_g|h\rangle\;=\;|gh\rangle, \qquad M_\chi|h\rangle \;=\;\chi(h)|h\rangle,\qquad h,g\in G,\;\chi\in\widetilde{G}$$
which satisfy the following relations
\begin{equation}\label{eq: elementary ribbon}
    T_g^n\;=\;{\bf 1}\;=\;M_\chi^n,\qquad T_gM_\chi\;=\;\chi(\bar g)M_\chi T_g
\end{equation}
 They define local elements $T^{(e)}_g$ and $M^{(e)}_\chi$  in $\A$ in the standard way, for any $e\in E.$
 For $\Lambda\in \Kk(E)$, a $G$-connection on $\Lambda$ is a map $c\colon \Lambda\to G.$ We  denote by $\mathfrak{C}_G(\Lambda)$ the set of all $G$-connections on $\Lambda$, and set
 \begin{equation*}
     \mathfrak{C}_G(E)\;:=\;\bigcup_{\Lambda\in \Kk(E)}\mathfrak{C}_G(\Lambda)
 \end{equation*}
One can define the unitary operators $T_c, M_{\tilde{c}}\in \A_\Lambda$  for $G$- and $\widetilde{G}$-connections $c$ and $\tilde{c}$ on $\Lambda$ by setting
 \begin{equation}
     T_c\;=\;\prod_{e\in \Lambda} T^{(e)}_{c(e)},\quad  M_{\tilde c}\;=\;\prod_{e\in \Lambda} M^{(e)}_{\tilde{c}(e)}
 \end{equation}
 An important observation is that these operators provide a presentation of $\A$:
\begin{proposition}\label{prop: generators}
One has
\begin{equation*}
  \A\;=\;C^*\big\{T_c, M_{\tilde{c}}\mid c\in \mathfrak{C}_G(E),\; \tilde{c}\in \mathfrak{C}_{\widetilde{G}}(E)\big\}. 
\end{equation*}
\end{proposition}
 The interactions of the model associated with the vertex $v\in V$ and face $\tilde{v}\in \tilde{V}$ are unitary operators from the local algebras $\otimes_{e \ni v} \Aa_e$ and $\otimes_{e \in \tilde v} A_e$, respectively, given by
\begin{equation}\label{eq: relations}
    A^g_v\;=\;\prod_{e\ni v} T^{(e)}_{g^{\zeta(e,v)}}, \quad B^\chi_{\tilde v} \;=\;\prod_{e\in \tilde v} M^{(e)}_{\chi^{\zeta(e,\tilde v)}},\qquad g\in G,\;\chi\in\widetilde{G}
\end{equation}
where $\zeta(e,v)=1$ if $e$ points away from $v$ and $-1$ otherwise. Similarly, $\zeta(e,\tilde v)=1$ if $e$ matches the orientation of $\tilde v$ and $-1$ otherwise. The operators $A_v^g$ and $B^\chi_{\tilde v}$ are known as the vertex and face operators, respectively, and their actions can be depicted graphically as follows
\[
    A_v^{g}\Star{g_2}{g_1}{g_3}{g_4}
\;=\;
\Star{g_2\bar g}{gg_1}{g_3\bar g}{gg_4},\qquad
B_{\tilde{v}}^{\chi}\,\Plaquette{g_1}{g_2}{g_3}{g_4}
\;=\;
\chi(\bar g_1g_2g_3\bar g_4)\Plaquette{g_1}{g_2}{g_3}{g_4}
\]
The commutativity of $G$ and the orientation related conventions in \eqref{eq: relations} imply
\begin{equation}\label{eq: commutative}
    [A_v^g,A_{v'}^h]\;=\;[A_v^g,B_{\tilde{v}}^\chi]\;=\;[B_{\tilde v}^\chi,B_{\tilde{v}'}^\mu]\;=\;0,\quad (A_v^g)^n\;=\; (B_{\tilde v}^\chi)^n\;=\;{\bf 1}
\end{equation}
From the operators $A_v^g$ and $B_{\tilde v}^\chi$, one constructs the commuting projections
\begin{equation}
    P_v\;:=\; \frac{1}{|G|}\sum_{g\in G} A_v^g,\qquad P_{\tilde v}\;:=\;\frac{1}{|\widetilde{G}|}\sum_{\chi\in \widetilde{G}}B_{\tilde{v}}^\chi
\end{equation}
and the net of Kitaev Hamiltonians 
\begin{equation}\label{eq: general Hamiltonian1}
   \Kk(E)\ni \Lambda \;\mapsto\; H_\Lambda \;=\;\sum_{v\dot \in V_\Lambda }({\bf 1}-P_v)+\sum_{\tilde{v}\dot \in \tilde{V}_\Lambda}({\bf 1}-P_{\tilde v}) \in \Aa_\Lambda,
\end{equation}
where $v \dot \in V_\Lambda$ indicates that the emanating edges of the vertex $v$ are all contained in $\Lambda$ and, similarly, $\tilde v \dot \in \tilde{V}_\Lambda$ indicates that the edges of the face $\tilde v$ are all contained in $\Lambda$. 

The ground state manifolds of these models are well known \cite{CNN}, and the excited states and spectra can be studied using the so-called ribbon operators. A \emph{ribbon} $\rho$ on $\Ll$ is a finite oriented path of $\Ll$ with no self-intersections. A ribbon is said to be open if its initial and terminal endpoints are distinct, and closed if they coincide. Similarly, we denote by $\tilde{\rho}$ the ribbons in the dual graph  $\tilde{\Ll}$. There is a natural $\Z_2$-valued pairing between ribbons and their edges, determined by orientation. Namely, for any edge $e\in \rho$ we define $\beta(e,\rho)$ as $1$ if the ribbon $\rho$ traverses $e$ in the same orientation as the lattice $\Ll$, and $\beta(e,\rho)=-1$  otherwise.  Similarly, for a ribbon in the dual lattice $\tilde{\rho}$, we define $\beta(e,\tilde{\rho})=\beta(\tilde{e},\tilde{\rho})$ where $\tilde{e}$ is the unique edge in $\tilde{\Ll}$ intersecting $e.$ Thus one can define the \emph{ribbon operators} as
\begin{equation}
    F_{\tilde \rho}^g\;:=\;\prod_{e\cap \tilde{\rho}\neq \emptyset }T^{(e)}_{g^{\beta(e,\tilde{\rho})}},\qquad  F_{ \rho}^\chi\;:=\;\prod_{e\in  \rho }M^{(e)}_{\chi^{\beta(e,\rho)}},
\end{equation}
The ribbon operators enter into the following relations with the star and plaquette operators 
\begin{equation}\label{Eq:RibbonRelations}
\begin{aligned}
& A_{v}^g F_\rho^\chi \;= \;\big(\prod_{\substack{e\ni v\\ e\in \rho}}\chi(g)^{\sign(e,\rho,v)}\big) F_\rho^\chi A_{v}^g ,   \quad  & F_{\tilde{\rho}}^g B_{\tilde v}^{\chi'} \;=\; \big(\prod_{\substack{e\in \tilde{v}\\ e\cap\tilde{\rho}\neq \emptyset }}\chi(g)^{\sign(e,\tilde{\rho},\tilde{v})}\big) B_{\tilde v}^{\chi'}F_{\tilde{\rho}}^g  
\end{aligned}
\end{equation}
where $\sign(e,\rho,v)=-\zeta(e,v)\beta(e,\rho)$ and, similarly, $\sign(e,\tilde \rho,\tilde v)=-\zeta(e,\tilde v)\beta(e,\tilde \rho)$. 

\subsection{The commutative algebra of local interactions}
The commutative algebra generated by the vertex and face operators will play a central role. Here is our first characterization of it:

\begin{proposition}
The Gelfand spectrum of the commutative $C^*$-algebra $$\Cc\;:=\;C^*\big\{ A_v^g, B_{\tilde v}^\chi\mid g\in G, \; \chi\in \widetilde{G},\; v\in V,\; \tilde{v}\in \tilde{V}\big \}$$ can be identified with the  Cantor space given by 
    $$\Omega\;=\;\Omega_V\times \Omega_{\tilde V}, \quad \Omega_V\;=\;\{ f\colon V\to \widetilde{G}\},\quad \Omega_{\tilde V}\;=\;\{ f\colon \tilde{V}\to G\}. $$
As a consequence, $\Cc \simeq C(\Omega)$.
\end{proposition}
\begin{proof}
The commutativity follows directly from the relations in \eqref{eq: commutative}.
To describe the spectrum, let $\omega$ be a character of $\Cc$, to which we associate the map 
\begin{equation}\label{eq: f1}
f_\omega \colon V \to \widetilde G, \quad [f_\omega(v)](g)\;=\;\omega(A_v^g).
\end{equation}
It is well defined because $g\mapsto A_v^g$ is a faithful unitary representation of $G$ and, as such, $g \mapsto \omega(A_v^g)$ is a character of $G$. Similarly, we have the well defined map
\begin{equation}\label{eq: f2}
\tilde f_\omega \colon \tilde V \to \widetilde{\widetilde G}\simeq G , \quad [\tilde f_\omega(\tilde v)](\chi)\;=\;\omega(B_{\tilde v}^\chi ).
\end{equation}
Since $A_v^g$ and $B_{\tilde v}^\chi$ generate $\Cc$, $\omega \mapsto (f_\omega,\tilde f_\omega)$ is an injective map from the spectrum of $\Cc$ to $\Omega$. Conversely, any element of $\Omega$ extends by linearity to a character of $\Cc$, and this supplies the inverse of $\omega \mapsto (f_\omega,\tilde f_\omega)$. Finally, since $G$ is non-trivial, $\Omega$ with the natural product topology is a Cantor space.
\end{proof}

\begin{proposition}\label{prop: classic interactions}
Under the identification $\Cc\simeq C(\Omega)$:
\begin{enumerate}[{\rm i)}]
    \item The vertex and face operators  $A_v^g$ and $B_{\tilde v}^\chi$ become the functions
$$A_v^g(\omega)\;=\;\langle \omega(v), g\rangle,\quad B_{\tilde v}^\chi(\omega)\;=\;\langle \chi, \omega({\tilde v})\rangle$$
where $\langle \cdot,\cdot\rangle \colon \widetilde{G}\times {G}\to \T$ is the natural pairing. 
\item Conjugations by $T^{(e)}_g$ leave $\Cc$ invariant and
\[
[{\rm Ad}_{T^{(e)}_g}(C)](\omega)=C(\delta_g^e \cdot \omega), \quad \forall \ C \in \Cc,
\]
where $\delta_g^e \in \Omega$ is explicitly given by 
\begin{equation}\label{eq: deltag}
\delta_g^e(\tilde v)=
\begin{cases}
g^{\zeta(e,\tilde v)} & \text{if } e \in \tilde v,\\
1_G & \text{otherwise},
\end{cases}
\qquad
\delta_g^e(v)=1_{\widetilde G}\ \ \text{for all }v\in V.
\end{equation}
\item Conjugations by $M^{(e)}_\chi$ leave $\Cc$ invariant and
\[
[{\rm Ad}_{M^{(e)}_\chi}(C)](\omega)=C(\delta_\chi^e \cdot \omega), \quad \forall \ C \in \Cc,
\]
where $\delta_\chi^e \in \Omega$ is explicitly given by 
\begin{equation}\label{eq: deltachi}
\delta_\chi^e(v)=
\begin{cases}
\chi^{\zeta(e,v)} & \text{if } e \ni v,\\
1_{\widetilde G} & \text{otherwise},
\end{cases}
\qquad
\delta_\chi^e(\tilde v)=1_{G}\ \ \text{for all } \tilde v\in \widetilde V.
\end{equation}
\end{enumerate}
\end{proposition}

\begin{proof}
    i) The statement follows directly from \eqref{eq: f1} and \eqref{eq: f2}. ii) It is enough to verify the statement on the generators of $\Cc$. In the standard presentation, the conjugation by the operator $T_g^{(e)}$ gives
\[
{\rm Ad}_{T_g^{(e)}}(B_{\tilde v}^\chi)
\;=\; \chi(g)^{\zeta(e,\tilde v)}\,B_{\tilde v}^\chi,
\]
for all $ \tilde v$ such that $e\in \tilde v$, while the conjugation acts trivially on the remaining generators of $\Cc$.
Therefore, when viewed as functions over $\Omega$,
\[
[{\rm Ad}_{T_g^{(e)}}(B_{\tilde v}^\chi)](\omega)
\;=\; \chi(g)^{\zeta(e,\tilde v)}\,\langle \chi, \omega(\tilde v)\rangle\;=\;\langle \chi,g^{\zeta(e,\tilde v)}\omega (\tilde v)\rangle
\]
if $e\in \tilde v$, and $[{\rm Ad}_{T_g^{(e)}}(X)](\omega)=X(\omega)$ for the rest of the generators of $\Cc$. The statement then follows. Similar for iii).
\end{proof}

To further investigate the properties of $\Cc$, we introduce generalizations of the Kitaev models. For this, note that the operators $A_v^g$ and $B_{\tilde v}^\chi$ generate additional commuting projections
\begin{equation}\label{eq: group projections}
    P_v^\chi\;:=\; \frac{1}{|G|}\sum_{g\in G}\chi(g) A_v^g,\qquad P_{\tilde v}^g\;:=\;\frac{1}{|\widetilde{G}|}\sum_{\chi\in \widetilde{G}}\chi(g)B_{\tilde{v}}^\chi.
\end{equation}
Setting $W := V \cup \widetilde V$, for each $\omega \in \Omega$,
we obtain a family of commuting projections $\big\{P_w^{\omega}\big\}_{w\in W}$ defined by $P_w^\omega:=P_w^{\omega(w)}$, and the net of  generalized \emph{Kitaev Hamiltonian} 
\begin{equation}\label{eq: general Hamiltonian}
   \Kk(E)\ni \;\Lambda\; \mapsto H_\Lambda^\omega\;:=\;\sum_{v\dot \in V_\Lambda }({\bf 1}-P_v^{\omega(v)})+\sum_{\tilde{v}\dot \in \tilde{V}_ \Lambda}({\bf 1}-P_{\tilde v}^{\omega({\tilde v})}) \in \Aa_\Lambda.
\end{equation}
In particular, we recover the standard net of Kitaev Hamiltonians when $\omega=\omega_0$, where $\omega_0$ denotes the character that is identically equal to the neutral element $1_G$ of  $G$ for all $v\in V$ and to the neutral element $1_{\widetilde{G}}$ of $\widetilde{G}$  at every dual vertex $\tilde{v}.$

\begin{remark} The generalized Hamiltonian \eqref{eq: general Hamiltonian} is interesting in its own right, but in this work, these generalizations are introduced mostly for technical reasons. $\hfill \blacktriangleleft $
\end{remark}

\begin{remark}\label{rem: hamiltonian}
Observe that $H_\Lambda^\omega \in \Cc$, hence it can be identified with a
classical observable in $C(\Omega)$. More precisely, identifying
$\Lambda$ with a finite subset of $W$ and using
Remark~\ref{prop: classic interactions}, we obtain
\[
H_\Lambda^\omega(\omega') \;=\;
\sum_{w\in \Lambda}
\delta_{\omega(w)}\big(\omega'(w)\big),
\qquad \omega'\in \Omega,
\]
since $P_w^\omega(\omega')=\delta_{\omega(w)}\big (\omega'(w)\big )$. In particular, the generalized Hamiltonian may be regarded as a continuous function $
H_\Lambda \colon \Omega\times \Omega \to \R_+.$ $\hfill \blacktriangleleft $
\end{remark}

The ground state projection of the Hamiltonian \eqref{eq: general Hamiltonian} can be written as 
\begin{equation}\label{eq: projections}
    P_\Lambda^\omega \;=\;\prod_{v\in V_\Lambda}P_v^{\omega(v)}\prod_{\tilde{v}\in \tilde{V}_ \Lambda}P_{\tilde{v}}^{\omega(\tilde{v})}\;=\; \prod_{w\in W} P_w^{\omega(w)},\qquad \Lambda\in \Kk(E)
\end{equation}
This defines a frustration-free net of projections $\{P_\Lambda(\omega)\}_\Lambda$ in the sense of \cite[Definition 2.2]{DET1}. A state $\SM$ on $\mathcal{A}$ is called a \emph{frustration-free (FF) ground state} for a frustration free net $\{Q_\Lambda\}$ of projections if $\SM(Q_\Lambda) = 1$ for all $\Lambda \in \Kk(E)$ \cite{DET1}. There is a direct relation between the space $\Omega$ and the FF ground states of the generalized models. Indeed, since $\Omega$ is the Gelfand spectrum of $\mathcal{C}$, each point $\omega \in \Omega$ defines a pure state on $\mathcal{C}$ via evaluation. 
Accordingly, we will henceforth identify points of $\Omega$ with states on $\mathcal{C}$. We have:

\begin{proposition}\label{prop: FF}
    Let $\omega\in \Omega$. Then any extension of the state $\omega$ over $\Cc$ to a state $\SM_\omega$ over $\A$ is a FF ground state for the net $\{P_\Lambda^\omega\}$.
\end{proposition}
\begin{proof}
      This is a consequence of the fact that all $P_\Lambda^\omega$ reside inside $\Cc$, hence $\SM_\omega (P_\Lambda^\omega)=\omega(P_\Lambda^\omega)$ for any extension $\SM_\omega$ of a state $\omega$ over $\Cc$. As such, $\SM_\omega(P_\Lambda^\omega) = P^\omega_\Lambda(\omega) =1$ (see remark \ref{rem: hamiltonian}), and the statement follows.
\end{proof}

\begin{definition}[\cite{DET1}]\label{Def:LTQO}
A net $\{Q_\Lambda\}$ of projections in $\mathcal{A}$ is said to satisfy the \emph{local topological quantum order} (LTQO) condition if, for every local observable $X \in \Aa_{\rm loc}$, one has
\[
\lim_{\Lambda} \| Q_\Lambda X Q_\Lambda - {\mathtt s}_\Lambda(X) Q_\Lambda \| = 0,
\]
where
\[
\SM_\Lambda(X) = \frac{\mathrm{Tr}(Q_\Lambda X Q_\Lambda)}{\mathrm{Tr}(Q_\Lambda)}
\quad \text{and} \quad X \in \Aa_\Lambda.
\]
\end{definition}

\begin{theorem}[\cite{DET1}]\label{Teo: unique FF GS}
If a frustration-free net of proper projections $\{Q_\Lambda\}$ satisfies the LTQO property, then the net converges to a minimal projection in the double dual of $\Aa$. Consequently, there exists a unique frustration-free ground state $\SM$, which is moreover pure, and it is explicitly given by the weak$^\ast$-limit \(\SM= \lim_{\Lambda} \SM_\Lambda.
\)
\end{theorem}

For the standard abelian Kitaev models, LTQO was proved in \cite[Theorem III.4]{Penneys}. We want to show that it also applies to the net of projections \eqref{eq: projections}. For this, we need the following concept:

 \begin{definition}
 Two nets of projections $\{Q_\Lambda\}$ and $\{Q'_\Lambda\}$ are \emph{locally equivalent} if there exists a net of locally representable automorphisms $\{\alpha_\Lambda \}$ of $\A$ such that $\alpha_\Lambda(Q_\Lambda)=Q_\Lambda'$ and $\alpha_\Lambda(\A_\Lambda)=\A_\Lambda$ for all $\Lambda\in \Kk(E).$ 
 \end{definition}
 
\begin{proposition}\label{prop: equivalence LTQO}
       The LTQO property is invariant under the local equivalence of nets of projections.
\end{proposition}
\begin{proof}
Assume that \(\{Q_\Lambda\}\) satisfies LTQO and let \(\{Q'_\Lambda\}\) be locally equivalent to \(\{Q_\Lambda\}\).
Thus, for each \(\Lambda\) there exists a locally representable automorphism \(\alpha_\Lambda\in\mathrm{Aut}(\mathcal A)\) such that
\[
\alpha_\Lambda(Q'_\Lambda)\;=\;Q_\Lambda
\qquad\text{and}\qquad
\alpha_\Lambda(\mathcal A_\Lambda)=\mathcal A_\Lambda.
\]
Fix \(X\in\mathcal A_{\mathrm{loc}}\). For \(\Lambda\) large enough we have \(X\in\mathcal A_\Lambda\), and we set
\(Y:=\alpha_\Lambda(X)\in\mathcal A_\Lambda\).
Since \(\alpha_\Lambda\) is a \(^*\)-automorphism, it is isometric, hence
\begin{align*}
\big\|Q'_\Lambda X Q'_\Lambda-\mathtt s'_\Lambda(X)\,Q'_\Lambda\big\|
&\;=\;\big\|Q_\Lambda Y Q_\Lambda-\mathtt s'_\Lambda(X)\,Q_\Lambda\big\|.
\end{align*}
It remains to identify \(\mathtt s'_\Lambda(X)\) with \(\mathtt s_\Lambda(Y_\Lambda)\).
Since \(\alpha_\Lambda(\mathcal A_\Lambda)=\mathcal A_\Lambda\), the restriction
\(\alpha_\Lambda|_{\mathcal A_\Lambda}\) is a \(^*\)-automorphism of the finite-dimensional
C\(^*\)-algebra \(\mathcal A_\Lambda\). Hence \(\alpha_\Lambda|_{\mathcal A_\Lambda}\) is inner, so it preserves
the trace \(\mathrm{Tr}\) on \(\mathcal A_\Lambda\). 
Therefore,
\[
\mathtt s'_\Lambda(X)
\;=\;\frac{\mathrm{Tr}(Q'_\Lambda X Q'_\Lambda)}{\mathrm{Tr}(Q'_\Lambda)}
\;=\;\frac{\mathrm{Tr}\!\big(\alpha_\Lambda(Q'_\Lambda X Q'_\Lambda)\big)}{\mathrm{Tr}\!\big(\alpha_\Lambda(Q'_\Lambda)\big)}
\;=\;\frac{\mathrm{Tr}(Q_\Lambda Y Q_\Lambda)}{\mathrm{Tr}(Q_\Lambda)}
\;=\;\mathtt s_\Lambda(Y_\Lambda).
\]
Combining the previous identities yields
\[
\big\|Q'_\Lambda X Q'_\Lambda-\mathtt s'_\Lambda(X)\,Q'_\Lambda\big\|
=
\big\|Q_\Lambda Y Q_\Lambda-\mathtt s_\Lambda(Y)\,Q_\Lambda\big\|.
\]
Since \(\{Q_\Lambda\}\) satisfies LTQO, the right-hand side converges to \(0\) as \(\Lambda\) increases.
Thus \(\{Q'_\Lambda\}\) satisfies LTQO. 
\end{proof}
The following lemma shows that the projections in \eqref{eq: group projections} transform under a family of automorphisms of $\A$ according to the left action of $G$ and $\widetilde{G}$. These automorphisms have been previously used in \cite{Naaij1,Naaij2,DE1}.
\begin{lemma}\label{lem: auto}
There is a family of locally representable automorphisms $$\big\{\alpha_w^g, \alpha_w^\chi\mid g\in G,\;\chi\in \widetilde{G},\;w\in W\big\}$$ of $\A$ which commutes pairwise and satisfies the following properties
\begin{enumerate}[{\rm i)}] 
    \item $\alpha_w^g(P_{w}^h)=P_w^{gh}$ and $\alpha_w^g(P_{w'}^h)=P_{w'}^h$ for any pair $w \neq w'$ from $W$ and $g,h\in G$.
    \item $\alpha_w^g(P_{w'}^\chi)=P_{w'}^\chi$ for all $w'\in W$ and $\chi\in \widetilde{G}$.
\end{enumerate}
Similar properties are satisfied by $\alpha_w^\chi$, which can be deduced from above by switching $G$ and $\widetilde{G}$. 
\end{lemma}
\begin{proof}
     Let $\zeta$ be a semi-infinite straight horizontal path of edges in $\Ll$, starting at vertex $v$ and progressing to the right, and let $\zeta_n$ denote the finite sub-path consisting of its first $n$ edges. For each element $X \in \Aa$ define
$$\alpha_v^\chi(X)\;=\;\lim_n F^\chi_{\zeta_n} X F^{\bar{\chi}}_{\zeta_n}.$$ Since the above limit converges in the operator norm, $\alpha_v^\chi$ defines a locally representable automorphism of $\Aa.$ Similarly, let $\tilde \zeta$ be a path on the dual graph $\tilde \Ll$, starting at the vertex $\tilde v$ and progressing to the right, and let $\tilde \zeta_n$ be the finite sub-path made up of its first $n$ edges. Then 
$$\alpha_{\tilde v}^g(X)\;=\;\lim_n F_{\tilde \zeta_n}^g X F^{\bar{g}}_{\tilde \zeta_n},$$
defines again a locally representable symmetry on $\Aa$. Then the stated actions of $\alpha_w$'s on $P_w$'s follow from \eqref{Eq:RibbonRelations}. Indeed, we have
\[
   F_\rho^\chi P_v^{\chi'}F_\rho^{\bar{\chi}}\;=\; P^{\chi'{\chi}^{\sign(e,\rho,v)}}_v,\qquad F_{\tilde{\rho}}^gP_{\tilde{v}}^hF_{\tilde \rho}^{\bar g}\;=\;P_{\tilde v}^{hg^{\sign(e,\tilde{\rho},\tilde{v})}} 
\]
when the ribbons $\rho$ and $\tilde \rho$ intersect $v$ and $\tilde{v}$ at exactly the edge $e$, as is the case here.
\end{proof}

\begin{theorem}\label{teo: unique gs}
    The frustration-free net of projections \eqref{eq: projections} satisfies the LTQO property. More precisely, for  any $X\in \Aa_\Lambda$ there exists $\Delta\supset \Lambda$ such that 
    $$P_\Delta^\omega X P_\Delta^\omega=\SM_\Delta(X)P_\Delta^\omega,\qquad \forall\,\omega\in \Omega$$
\end{theorem}
\begin{proof}
 The case $\omega =\omega_0$ was proved in \cite[Theorem III.4]{Penneys}. Thus, by Proposition \ref{prop: equivalence LTQO}, to conclude the proof, it is enough to show that the net of projections $\{ P_\Lambda^\omega\}$ is pairwise locally equivalent to the net $\{P_\Lambda^{\omega_0}\}$, for all $\omega\in \Omega$. By \eqref{eq: projections} and Lemma \ref{lem: auto}, by composing a finite number of $\alpha_w$'s in Lemma \ref{lem: auto}, one gets a locally representable automorphism $\alpha_\Lambda$ such that $\alpha_\Lambda(P_\Lambda^\omega)=P_\Lambda^{\omega_0}$. This completes the proof.
\end{proof}

\begin{definition}
    A sub-$C^\ast$-algebra $\Bb \subset \Aa$ is said to have the unique extension property (UEP) if any pure state over $\Bb$ extends uniquely over $\Aa$ as a pure state.
\end{definition} 

\begin{corollary}
    $\Cc$ has the UEP.
\end{corollary}
\begin{proof}
    Proposition~\ref{prop: FF} assures us that any extension $\SM_\omega$ of a pure state $\omega$ of $\Cc$ is a FF ground state for the net \eqref{eq: projections}. Furthermore, Theorems~\ref{Teo: unique FF GS} and \ref{teo: unique gs} then assure us that such extensions, which always exist, are unique and pure.
\end{proof}

\subsection{The $C^\ast$-diagonal of an Abelian Kitaev model}

 The normalizer set of a $C^\ast$-inclusion $\Bb \subset \Aa$ is defined as
\[
\mathscr{N}(\mathcal{B})
\;:=\;
\big\{
V \in \mathcal{A}
\;\mid\;
V B V^* \in \mathcal{B}
\ \text{and}\
V^* B V \in \mathcal{B}
\ \text{for all } B \in \mathcal{B}
\big\}.
\]
It follows from the very definition that $\mathcal{B}\subset \mathscr{N}(\Bb)$. The sub-algebra $\Bb$ is said to be regular if $\mathscr{N}(\Bb)$ is dense in $\A$, \ie $C^*( \mathscr{N}(\Bb))=\A.$  Furthermore:

\begin{definition}[\cite{Kum}]\label{def: diagonal} An abelian sub-$C^\ast$-algebra $\Bb$ of $\Aa$ is said to be a $C^\ast$-diagonal if it is regular and has the UEP.
\end{definition}

Definition~\ref{def: diagonal} differs slightly from the original one in \cite{Kum}, 
since it does not explicitly assume maximal abelianness nor the existence of a 
unique conditional expectation. However, as $\mathcal{A}$ is separable, both definitions actually coincide in our setting. Indeed,
UEP immediately implies that $\Bb$ is a maximal abelian sub-$C^\ast$-algebra (masa) of $\A$ \cite[Corollary 2.7]{AB2}. Moreover, it is also true that there exists a unique conditional expectation $E\colon \A\to \Bb$, which is uniquely determined by the relation
\begin{equation}\label{eq: pure extension}
    E(X)(\omega)\;=\; \SM_\omega(X)
\end{equation}
where we used the identification $\Bb\simeq C(P(\Bb))$, with $P(\Bb)$ the topological space of pure states of $\Bb$,  $\omega\in P(\Bb)$ and $\SM_\omega$ the unique pure extension of $\omega$ to $\A$ (see \cite[theorem 3.4]{And1}).
\begin{theorem}\label{Teo: diagonal}
     The algebra $\Cc$ is a $C^*$-diagonal of $\A$. Consequently, $\Cc$ is a masa of $\A$ with a unique conditional expectation $E\colon \A\to \Cc$.
\end{theorem}
\begin{proof}
We already showed that $\Cc$ has the UEP. Moreover, $\Cc$ is regular since the operators $T_c$ and $M_{\tilde c}$ normalize $\Cc$ for any connections $c$ and $\tilde{c}$, and  they generate $\A$ according to Proposition \ref{prop: generators}. Thus, $\mathcal{C}$ is a $C^\ast$-diagonal of $\mathcal{A}$, and the remainder of the proof follows from the discussion preceding this theorem.
\end{proof}

\section{The Weyl Groupoids of Abelian Kitaev Models}
\label{Sec:WGroupoid}

We start by recalling the definition of the Weyl groupoid associated wit a $C^\ast$-inclusion, based on \cite{Kum,Ren1}. Then we give our definition of the Weyl groupoid canonically associated to an Abelian Kitaev model, together with an explicit characterization of it.

A twist $\mathfrak{T}$ of a locally compact, Hausdorff, étale, topologically principal groupoid $\G$ with unit space $\G^{0}$, is a central extension of groupoids
\[
\mathbb{T}\times \G^{0}
\;\longrightarrow\;
\mathfrak{T}
\;\xrightarrow{\;\;}
\G,
\]
where $\mathbb{T}$ acts freely on $\mathfrak{T}$ so that $\mathfrak{T}/\mathbb{T}\cong \G$.
The twist $\mathfrak{T}$ is said to be \emph{trivial} if it is isomorphic to $\G\times\mathbb{T}$.  Given a $C^\ast$-diagonal inclusion, or more generally, a Cartan inclusion $(\A,\Bb)$, the fundamental results of \cite{Kum,Ren1} state that there exists a groupoid $\G$, called the Weyl groupoid of the pair, together with a twist $\mathfrak{T}$ of $\G$ such that
\[
(\A,\Bb)\;\cong\;\bigl(C_r^*(\G,\mathfrak{T}),\, C_0(\G^{0})\bigr).
\]
Moreover, the twist $\mathfrak{T}$ is uniquely determined by $(\A,\Bb)$ up to conjugation. In other words, the twist is a complete invariant of the $C^\ast$-inclusion.

Theorem \ref{Teo: diagonal} motivates the following definition.
    \begin{definition}
    The \emph{Weyl groupoid} of an Abelian Kitaev model is the Weyl groupoid $\G_\Cc$ of the $C^\ast$-inclusion $(\Aa,\Cc).$
\end{definition}

Our immediate task is to obtain an explicit description of $\G_\Cc$. Note that $\Omega$ has a natural group structure, and we let $\Sigma$ be the abelian subgroup of $\Omega$ of all configurations with finite support, \ie all elements in $\Omega$  that differ from the neutral configurations $1_{\widetilde{G}}$ and $1_G$ at only
finitely many sites. We equip $\Sigma$ with the final topology, in which case it becomes a locally compact topological group. Consider the corresponding transformation groupoid $\G_\Sigma:=\Omega\rtimes_\eta \Sigma$ where $\eta$ stands for the action by translations, \ie
$$\Omega\times \Sigma\ni(\omega,\sigma)\;\mapsto \;\eta_\sigma(\omega)\;:=\;\sigma \cdot \omega.$$
There is a natural surjective homomorphism $\hat{\pi}\colon \Sigma \to \widetilde G\times G$ defined by 
$$\hat{\pi}(\sigma)=\big(\prod_{v\in V\cap {\rm supp}(\sigma)} \sigma(v), \prod_{\tilde{v}\in \tilde{V}\cap {\rm supp}(\sigma)}\sigma(\tilde{v})\big). $$
The above products are well-defined because any element in $\Sigma$ has finite support.  Moreover, $\hat \pi$ can be lifted to a groupoid morphism on $\G_\Sigma$ by setting $\pi(\omega,\sigma)=\hat{\pi}(\sigma)$. Denote by $\Gamma:={\rm Ker}(\hat{\pi})$ and $\G_\Gamma:=\Omega\rtimes_\eta \Gamma={\rm Ker}(\pi)$. It follows the groupoid exact sequence
  \begin{equation}\label{eq: groupoid extension}
      1 \longrightarrow \G_\Gamma
\;\xrightarrow{\;\;}
\G_{\Sigma}
\;\xrightarrow{\;\pi\;}
\widetilde G \times G
\longrightarrow 1
  \end{equation}

\begin{proposition}\label{prop: automorphism ribbon}
The action of $\Gamma$ on $\Omega$ is free and minimal. Moreover, $\eta(\Gamma)$ coincides with the subgroup of ${\rm Homeo}(\Omega)$ generated by conjugations with ribbon operators.
\end{proposition}

\begin{proof}
The action of $\Gamma$ on $\Omega$ is free, as it is given by right translations in the group $\Omega$. Furthermore, it is also minimal since the subgroup $\Gamma$ is dense in $\Omega$: every non-empty cylinder subset of $\Omega$ has a non-trivial intersection with $\Gamma$.

Denote by $\mathcal{R}$ the set of all ribbon operators in $\A$, and consider the corresponding action given by conjugation  ${\rm Ad}\colon \mathcal{R}\to {\rm Homeo}(\Omega)$. For $g\in G$, we have seen in Proposition~\ref{prop: classic interactions} that the conjugation by the elementary ribbon operator $T_g^{(e)}$ is induced from the translation by $\delta_g^e\in \Gamma$ on $\Omega$, spelled out in \eqref{eq: deltag}. This $\delta_g^e$ belongs to the kernel of $\hat \pi$ because the edge $e$ always belongs to two and only two faces, for which the coefficients $\zeta$ in \eqref{eq: deltag} are opposite. A similar argument applies for the conjugation with $M_\chi^e$. In this case, it is induced from the translation by $\delta_\chi^e$ from \eqref{eq: deltachi}. The latter also belongs to the kernel of $\hat \pi$ because the edge $e$ is always shared by two and only two vertices, for which the coefficients $\zeta$ in \eqref{eq: deltachi} are opposite. Hence, the elementary ribbon operators belong to $\eta(\Gamma)$. A general ribbon operator is the product of the elementary ribbons with an orientation for which the action is only non-trivial on the initial and final vertices of the paths. Thus, the analysis is similar and one obtains that  ${\rm Ad}(\mathcal{R})\subset \eta(\Gamma)$. 

Conversely, let $\gamma\in\Gamma$ and set
$F=\supp(\gamma)\cap V$, which is finite. Assume $|F|>1$, hence the case $|F|\neq 1$. 
Fix a base vertex $v_\ast\in F$, and
for each $v\in F\setminus\{v_\ast\}$ choose an oriented ribbon
$\rho=(e_1,\dots,e_m)$ connecting $v$ to $v_\ast$ so that $\beta(e_1,\rho)=1$. Define
\[
\gamma_v\;:=\;\prod_{e\in \rho}\delta_{\gamma(v)^{\beta(e,\rho)}}^{e}
\]
By construction, $\gamma_v(v)=\gamma(v)$, it is trivial at every vertex different from $v$ and $v_*$, and at the base vertex we have $\gamma_v(v_*)=\gamma(v)^{-1}$. Taking the product over all $v\in F\setminus\{v_\ast\}$,
we obtain
$\gamma|_V=
\prod_{v\in F\setminus\{v_\ast\}}
\gamma_v,$ 
where we used the condition $\prod_{v\in F}\gamma(v)=1_{\widetilde{G}}$ that
ensures the cancellation at the base vertex $v_*$. Fixing also a face $\tilde v_*\in \tilde{F}={\rm supp}(\gamma)\cap \tilde{V}$ and repeating a similar construction, one obtains
$$\gamma\;=\;
\prod_{v\in F\setminus\{v_\ast\}}
\gamma_v\prod_{\tilde{v}\in \tilde{F}\setminus\{\tilde{v}_\ast\}}
\gamma_{\tilde v}$$
Therefore $\Gamma$ is generated by the elements
$\delta_\chi^e$ and $\delta_g^e$, and one concludes that ${\rm Ad}(\mathcal{R})=\eta(\Gamma)$.
\end{proof}
\begin{remark}
    The space $\Omega$ consists of the unique frustration-free ground state $\omega_0$ of the standard Kitaev Hamiltonian together with all its excitations, including configurations with infinitely many violations. Proposition \ref{prop: automorphism ribbon} shows that two pure states $\omega_1$ and $\omega_2$ in $\Omega$ can be connected by a finite product of ribbon operators if they differ by an element $\Gamma$.  In particular, the orbit of $\omega_0$ under the action of $ \Gamma$ precisely describes the states in the trivial sector of the Kitaev Hamiltonian. All this information is encoded by the groupoid $\G_\Gamma$, whose elements are equivalence classes of pairs $(\omega_1,\omega_2)$ related by conjugations with ribbon operators. $\hfill \blacktriangleleft $
\end{remark}
\begin{remark} 
The sequence \eqref{eq: groupoid extension} shows that $\G_\Sigma$ decomposes into a ``gauge'' part
$\G_\Gamma$ and a global symmetry part $G\times \widetilde G$.
The groupoid $\G_\Gamma$ encodes local charge--flux configurations on $\Omega$
with trivial total anyonic content, whereas $G\times \widetilde G$
keeps track of single excitations.
Equivalently, the subgroup $\Gamma$ generates local anyonic excitations whose
net topological charge vanishes, that is, collections of anyons that mutually
annihilate. $\hfill \blacktriangleleft $
\end{remark}

\begin{theorem}
    $\G_\Gamma$ is isomorphic to the Weyl groupoid $\G_\Cc$ of the Cartan pair $(\A,\Cc)$.
\end{theorem}
\begin{proof}
   Observe that $\Gamma$ is a locally finite group. Consequently, \cite[Theorem~3.8]{GTS} and Proposition \ref{prop: automorphism ribbon} implies that $\G_\Gamma$ is a $AF$-relation in the sense of \cite[Definition~3.7]{GTS}. Therefore, $C^*_r(\G_\Gamma)$ is an $AF$-algebra. Since the twist is trivial on $AF$-relations \cite[Proposition 4.6]{DE1}, by \cite[Proposition 4.13]{Ren1} and Proposition \ref{prop: automorphism ribbon}, to complete the proof, it is enough to show that $C^*_r(\G_\Gamma)\simeq \A$. By Krieger’s result \cite{Krieger} (see also \cite[Theorem~4.10]{Mat}),  the latter  is equivalent to compute the ordered dimension group $(H_0(\G_\Gamma),H_0(\G_\Gamma)^+,[1_{\Omega}])$, as defined in  \cite[Definition~3.1]{Mat}, and verify that
    \[
(H_0(\G_\Gamma),H_0(\G_\Gamma)^+,[1_{\Omega}])\;\simeq\; (\Z[1/n],\Z_+[1/n],1)
\]
Consider the  Bernoulli probability measure $\mu$ on $\Omega$, \textit{i.e.}, $\mu=\big(\bigotimes_V\delta)\otimes \big(\bigotimes_{\tilde{V}}\tilde{\delta}\big)$ with $\delta$ and $\tilde{\delta}$ the Haar measures on $G$ and $\widetilde{G}$, respectively. It is $\Gamma$-invariant, and therefore the map
\[
\varphi\colon H_0(\G_\Gamma)\to \mathbb{R},
\qquad
\varphi([Q])=\int_\Omega Q \,\mathrm{d}\mu,
\]
is a group homomorphism \cite[Section~6]{Mat}. By emulating the argument in the proof of Lemma 4.4 in \cite{DE1} one can check $\varphi$ is injective, and its image coincides with $\Z[1/n]$. Moreover, $\varphi$ preserves the order dimension, concluding the proof.
\end{proof}

% Since $\Cc$ is a $C^*$-diagonal of $\A$, then there is a unique conditional expectation $E\colon \A\to \Cc$ which can be described as follows: for any $X \in \mathcal{A}$, the compression $E(X)$ is determined by the  continuous function 
% $E(X)(\omega) = \SM_\omega(X),$ where $\omega\in \Omega$ and  $\SM_\omega$ denotes the unique pure extension of the state $\omega$ on $\mathcal{C}$ to $\mathcal{A}$.

\medskip

\section{The complete set of KMS states}
This section is devoted to the explicit computation of the KMS states for the abelian Kitaev model. We begin by recalling the definition of KMS states, along with the notion of KMS measures in the setting of abstract groupoids.
\subsection{KMS states}
Denote by $\alpha^H_\omega$ the dynamics on $\A$ generated by the generalized Kitaev Hamiltonian \eqref{eq: general Hamiltonian}.  We will focus on the case $\omega=\omega_0$, as the general case can be treated in the same way, and we shall write the corresponding dynamics as $\alpha^H$ for short.  To describe the set of ${\rm KMS}$ states on the dynamical system $\big(\A,\alpha^{H}, \R\big)$ let us recall its standard definition \cite{bratteli-robinson-97}. We say that a  $\alpha^H$-invariant state 
 $\SM$ for $\A$  is a $(\alpha^H,\beta)$-${\rm KMS}$ state at inverse temperature  $\beta\in [0,\infty)$ if for all $X$ and $Y$ in $\A$ there exists a function $F$ which  is continuous and bounded on the strip $0\leq {\rm Im} z\leq \beta $ and analytic on $0<{\rm Im}z<\beta$ so that meets the  following conditions for all $t\in \R$
\begin{enumerate}[i.]
    \item $F(t)=\SM(X\alpha^H_t(Y))$;
    \item $F(t+\ii \beta)=\SM(\alpha^H_t(Y)X)$
\end{enumerate}
Formally, the latter is equivalent to saying that for any $X$ and $Y$ in a dense set of entire analytic elements of $\alpha^H$ the following equality holds
\begin{equation}\label{eq: KMS}
    \SM\big( X\alpha^H_{\ii\beta}(Y)\big)=\SM\big(YX\big)
\end{equation}
when the expression $\alpha^H_{\ii\beta}(b)$ makes sense. The idea of this section is to find all the states $\SM$ that solve the equation \eqref{eq: KMS} for the dynamics induced by the Kitaev Hamiltonian \eqref{eq: general Hamiltonian1}. We shall see that such an equation is completely determined by suitable measures on $\Omega$ by using the Weyl groupoid.

\subsection{KMS measures}
 Consider an abstract amenable locally compact groupoid $(\G,\G^{0})$   with Haar system $\lambda:=\{\lambda_u\}_{u\in \G^{0}}$. We shall denote by $\lambda^{-1}:=\{\lambda^u\}_{u\in \G^{0}}$ the pushforward Haar system induced by the inversion map on $\G$.
A Borel measure $\mu$ on $\G^{0}$ induces a measure on $\G$ by setting
$$ \lambda_\mu(f)\;:=\;\int_{\G}f(\zeta)\,{\rm d}\mu(u){\rm d}\lambda_u(\zeta),\qquad f\in C_c(\G)$$
Similarly, one can also define a measure $\lambda_\mu^{-1}$ on $\G$ associated with the Haar system $\lambda^{-1}:=\{\lambda^u\}_{u\in \G^{0}}$.

\begin{definition}
    A Borel measure $\mu$ on $\G^{0}$ is said to be quasi-invariant if  $\lambda_\mu\sim \lambda_\mu^{-1}$, \ie they are mutually absolutely continuous. If this is the case, there is a \emph{Radon-Nikodym} derivative $\Delta_\mu:={\rm d} \lambda_\mu/ {\rm d}\lambda_
    \mu^{-1}$  called the \emph{modular function}, which defines a groupoid homomorphism $\Delta_\mu\colon \G \to \R_+^{\times}$, where $\R_+^{\times}$ denotes the multiplicative group of positive real numbers.
\end{definition}

\begin{remark}
Observe that an invariant measure on $\mathcal{G}^{0}$ is nothing but a quasi-invariant such that $\Delta_\mu(\zeta)=1$ for all $\zeta\in \mathcal{G}.$ $\hfill \blacktriangleleft $
\end{remark}
It is well known that any quasi-invariant measure $\mu$ induces a state $\SM_\mu$ on the groupoid $C^*$-algebra $C^*_r(\G)$ as follows
\begin{equation}\label{eq: induced state}
    \SM_\mu(f)\;:=\;\int_{\G^{0}} E(f)(u){\rm d}\mu(u)
\end{equation}
where recall that $E\colon C^*_r(\G)\to C_0(\G^{0})$ is the standard conditional expectation, given by restriction.

\medskip 

 A (continuous) \emph{$1$-cocycle}
with values in an abelian topological group $A$ is a (continuous) map
$c\colon\G \to A$ such that
\[
c(\zeta_1\zeta_2)\;=\;c(\zeta_1)+c(\zeta_2)
\qquad\text{for all } (\zeta_1,\zeta_2)\in \G^{(2)}.
\]
We denote by $Z^{1}(\G,A)$ the abelian group of continuous $1$-cocycles
$\G\to A$. If $c\in Z^{1}(\G,A)$ and $\widetilde A$ is the Pontryagin dual of $A$, then 
\[
\tilde a \;\mapsto\; [\alpha^c_{\tilde a}(f)](\zeta)\;:=\;\langle \tilde a,c(\zeta)\rangle f(\zeta), \quad f\in C_c(\G),
\]
extends to a $C^\ast$-dynamical system $(C^\ast_r(\G),\widetilde A,\alpha^c)$ \cite[Proposition~5.1]{RenaultBook}. In particular, if $A=\R$, the map
 $$(\alpha^c_tf)(\zeta)\;:=\;e^{\ii t c(\zeta)} f(\zeta),\qquad t\in \R, \ f\in C_c(\mathcal{G})$$
defines a strongly continuous family of $*$-automorphisms $t\mapsto \alpha^c_t$ on $C^*_r(\mathcal{G})$. Furthermore, $C_c(\G)$ consists of entire analytic elements for $\alpha^c$. The (pre-)generator of $\alpha^c$ is the $*$-derivation  $\delta^c(f)(\zeta):=\ii c(\zeta)f(\zeta)$. For $c\in Z^1(\G,\T)$, we can also define an automorphism $\alpha^c$ of $\A$ by setting
 $$(\alpha^cf)(\zeta)\;:=\; c(\zeta) f(\zeta),\qquad f\in C_c(\mathcal{G}).$$
\begin{definition}\cite{RenaultBook}
    Let $c\in Z^{1}(\G,\R)$ and $\beta\in\R$.
A probability measure $\mu$ on $\G^{0}$ is said to satisfy the
\emph{$(c,\beta)$-KMS condition} if $\mu$ is quasi-invariant and its modular function
satisfies
\[
\Delta_\mu(\zeta)\;=\;e^{-\beta\, c(\zeta)}
\qquad\text{for $\lambda_\mu$-a.e. }\zeta\in\G;
\]
\end{definition}

Any KMS measure induces a KMS state on the $C^*$-algebra $C^*_r(\G)$ in the standard way.
\begin{proposition}[Prop.~5.4, \cite{RenaultBook}]
    Let $c\in Z^{1}(\G,\R)$, $\beta\in [0,\infty]$, and $\mu$ a \emph{$(c,\beta)$-KMS} measure on $\G^{0}$.  Then the induced state $\SM_\mu$ according to \eqref{eq: induced state} is a $(\alpha^c,\beta)$-KMS state for $C^*_r(\G).$
\end{proposition}
It should be mentioned that a KMS state $\SM$ on $C^*_r(\G)$ also induces a KMS measure on $\G^{0}$ by restricting the state to the sub-$C^*$-algebra $C_0(\G^0)$ of $C^*_r(\G)$. This correspondence is bijective for principal groupoids \cite[Chapter 2, Proposition 5.4]{RenaultBook}, \ie any KMS state arises from a KMS measure on $\G^0$.  For étale groupoids, the principal assumption can be weakened, and the above result still holds \cite[Proposition 3.2]{KR1}. However, in the non-principal étale case, KMS measures alone are not sufficient to describe all KMS states on $C^*_r(\G)$. Namely, KMS states can instead be characterized in terms of fields of traces on the $ C^*$-algebras associated with the isotropy groups of $\G$ \cite{Nesh1}.

\subsection{The set of KMS states}
Coming back to the case $\G_
\Cc$, since it is a transformation groupoid, then a Borel measure $\mu$ on $\Omega$ is {quasi-invariant} if 
$\mu \sim \gamma_\ast \mu$
 for all $\gamma \in \Gamma$
where  $\gamma_\ast \mu(E) = \mu(\eta_\gamma^{-1} (E))$.  Here, the modular function takes the form
\[
\Delta_\mu(\omega,\gamma)
\;=\;
\frac{{\rm d}(\gamma_\ast \mu)}{{\rm d}\mu}(\omega),
\qquad (\omega,\gamma)\in \Omega \rtimes \Gamma .
\]
In particular, a quasi-invariant measure $\mu$ induces a state $\SM_\mu$ on $\A$ as follows
\begin{equation}\label{eq: induced state2}
    \SM_\mu(X)\;:=\;\int_\Omega E(X)(\omega){\rm d}\mu(\omega)
\end{equation}
where recall that $E\colon \A\to \Cc\equiv C(\Omega)$ is the unique conditional expectation, which in the groupoid presentation of $\Aa$ reduces to the restriction of $C^\ast_r(\G_\Cc)$ to $C(\G_\Cc^{0})$.   

\medskip 

With these ingredients in place, we are now in a position to show that the KMS states of the Kitaev Hamiltonian are in bijective correspondence with the KMS probability measures on $\Omega$.
\begin{theorem}\label{theo: KMS measures} The following hold:
\begin{enumerate}[{\rm i)}]
    \item There exists a cocycle $c\in Z^1(\G_\Cc,\R)$ such that $\alpha^H = \alpha^c$.
    \item All $(\alpha^H,\beta)$-KMS states are induced from $(\alpha^c,\beta)$-${\rm KMS}$ measures and every $(\alpha^c,\beta)$-${\rm KMS}$ measure induces a $(\alpha^H,\beta)$-KMS state.
\end{enumerate}

%Let $\beta\in \R$ be a finite temperature. A state  $\SM$ of $\A$ is a $(\alpha^H, \beta)$-${\rm KMS}$  state if and only if there exists $c\in Z^1(\G_\Cc,\R)$ and $(\alpha^c,\beta)$-${\rm KMS}$ measure $\mu$ on $\Omega$  such that $\SM=\SM_\mu.$
\end{theorem}

\begin{proof}
i) Observe that the dynamics $t\mapsto \alpha^H_t$ restricts to the identity on
$\Cc=C(\G_\Cc^0)$ for all $t\in\R$.
Therefore, since $\G_\Cc$ is a  principal etal\'e groupoid, by \cite[Corollary~3.2.6]{KomuraDocMath2025}, each automorphism
$\alpha^H_t$ is implemented by a unique $\T$-valued $1$-cocycle
$c_t\in Z^1(\G_\Cc ,\T)$. Moreover, since $\alpha^H$ is a strongly continuous action of\/ $\R$,
the map $(t,\omega,\gamma)\mapsto c_t(\omega,\gamma)$
is jointly continuous, and as such for each arrow $(\omega,\gamma)\in\G_\Cc$ the function
$t\mapsto c_t(\omega,\gamma)$ defines a continuous group homomorphism
from $\R$ to $\T$.
It follows that there exists a unique real number $c(\omega,\gamma)\in\R$
such that
\[
c_t(\omega,\gamma)\;=\;e^{\ii t\,c(\omega,\gamma)},
\qquad\text{for all }t\in\R.
\]
By joint continuity, we obtain a continuous real-valued cocycle
$c\in Z^1(\G_\Cc,\R)$ satisfying $\alpha^H_t=\alpha^c_t$ for all $t\in\R$. ii) Since $\G_\Cc$ is topologically principal \cite[Proposition 5.11]{Ren1}, this is a consequence of Proposition~5.4 in \cite[Chapter~II]{RenaultBook}.
\end{proof}

Now, the next step is to compute all the KMS measures for any finite temperature $\beta.$ Consider the measure $\mu_\beta$ on $\Omega$ by
\begin{equation}\label{eq: KMS measure}
    \mu_\beta\;:=\;\big(\bigotimes_{v\in V}\nu_\beta\big) \otimes  \big(\bigotimes_{\tilde{v}\in \tilde{V}}\tilde{\nu}_\beta\big) 
\end{equation}
where $\nu_\beta$ and $\tilde{\nu}_\beta$ are probability measures on $G$ and $\widetilde{G}$, respectively, which satisfies the following 
\[
\nu_\beta(\chi)
\;:=\;
\begin{cases}
\dfrac{e^{\beta}}{e^{\beta}+(|G|-1)}, & \text{if } \chi = 1_{\widetilde G}, \\[8pt]
\dfrac{1}{e^{\beta}+(|G|-1)}, & \text{if } \chi \neq 1_{\widetilde G},
\end{cases} 
\]
$$\tilde{\nu}_\beta(g)
\;:=\;
\begin{cases}
\dfrac{e^{\beta}}{e^{\beta}+(|G|-1)}, & \text{if } g = 1_G, \\[8pt]
\dfrac{1}{e^{\beta}+(|G|-1)}, & \text{if } g \neq 1_G,
\end{cases}$$
Since the above measures have full support, it follows that $\nu_\beta$ and $\tilde{\nu}_{{\beta}}$ are $G$-$\widetilde{G}$-quasi-invariants, where the Radon-Nikodym derivatives can be explicitly described as
\begin{equation*}
\frac{d(\nu_\beta\circ\tau_\chi^{-1})}{d\nu_\beta}(\chi')
\;=\;
\frac{\nu_\beta(\chi'\chi^{-1})}{\nu_\beta(\chi')},
\qquad
\frac{d(\tilde\nu_\beta\circ\tilde\tau_h^{-1})}{d\tilde\nu_\beta}(g)
\;=\;
\frac{\tilde\nu_\beta(gh^{-1})}{\tilde\nu_\beta(g)}.
\end{equation*}
with $\tau$ and $\tilde{\tau}$ the action by translation of $G$ and $\widetilde{G}$, respectively. Therefore, $\mu_\beta$ is also quasi-invariant, and its Radon-Nikodym  derivative can be written as
\begin{equation*}\label{eq:RN_mu_beta}
\begin{split}
  \frac{d(\gamma_*\mu_\beta )}{d\mu_\beta}(\omega)
\;&=\;
\prod_{v\in \supp(\gamma)\cap V}
\frac{\nu_\beta\!\big(\omega(v)\gamma(v)^{-1}\big)}
     {\nu_\beta\!\big(\omega(v)\big)}
\cdot
\prod_{\tilde v\in \supp(\gamma)\cap \widetilde V}
\frac{\tilde\nu_\beta\!\big(\omega(\tilde v)\gamma(\tilde v)^{-1}\big)}
     {\tilde\nu_\beta\!\big(\omega(\tilde v)\big)} \\
     &=\;e^{-\beta c_H(\omega,\gamma)}
\end{split}
\end{equation*}
where $c_H\colon \G_\Cc\to \R$ is the function given by
\begin{equation}\label{eq: cocycle}
    \begin{split}
        c_H(\omega,\gamma)\;=\;&\sum_{v\in \supp(\gamma)\cap V}\delta_{1_{\widetilde{G}}}\big(\omega(v)\big)-\delta_{1_{\widetilde{G}}}\big(\omega(v)\gamma(v)^{-1}\big)\\
        &+\sum_{\tilde{v}\in \supp(\gamma)\cap \tilde{V}}\delta_{1_G}\big(\omega(\tilde{v})\big)-\delta_{1_G}\big(\omega(\tilde{v})\gamma(\tilde{v})^{-1}\big)
    \end{split}
\end{equation}
A direct computation shows that $c_H(\omega,\gamma_1\gamma_2)=c_H(\omega,\gamma_1)+c_H(\eta_{\gamma_1^{-1}}(\omega),\gamma_2)$ and, as such, $c_H\in Z^1(\G_\Cc,\R)$. In particular, this cocycle implements the dynamics of the Kitaev model.
\begin{proposition}
The Kitaev dynamics $\alpha_t^H$ coincides with the one generated by the cocycle $c_H,$ \ie, $\alpha^H_t=\alpha_t^{c_H}.$ 
\end{proposition}
\begin{proof}
Recall that $\A\simeq C^*_r(\G_\Cc)\simeq C(\Omega)\rtimes \Gamma$. With respect to the crossed product structure, any $X\in \A$ admits a Fourier decomposition 
$$X\;=\;\sum_{\gamma\in \Gamma}X_\gamma F_\gamma,\qquad X_\gamma\in C(\Omega)$$
where $\gamma\mapsto F_\gamma$ is a unitary representation of $\Gamma.$ Hence, $X(\omega,\gamma)=X_\gamma(\omega)$ and formally 
$$\alpha_t^H(X)\;=\;\sum_{\gamma\in \Gamma}X_\gamma \alpha_t^H(F_\gamma)$$
Thus, it suffices to compute $\alpha_t^H(F_\gamma)$. Let $H_\Lambda\in C(\Omega)$ be the local Hamiltonian associated with a finite region
$\Lambda\subset W$. Since conjugation by $F_\gamma$ implements the action of $\gamma$ on $\Cc$, using remark \ref{rem: hamiltonian} one gets
\[
\big(F_\gamma^* H_\Lambda F_\gamma - H_\Lambda\big)(\omega)
\;=\;c_{H}(\omega|_\Lambda,\gamma){\bf 1}
\]
Since $\gamma$ has finite support, for $\Lambda$ large enough
$c_{H}(\omega|_\Lambda,\gamma)=c_H(\omega,\gamma)$ for all $\omega\in\Omega$. Using this identity, we obtain $
\alpha_t^H(F_\gamma)
=
e^{it c_H(\cdot,\gamma)}\,F_\gamma.$
Consequently, we formally have
\[
\alpha_t^H(X)\;=\;
\sum_{\gamma\in\Gamma} X_\gamma\, \alpha_t^H(F_\gamma)
\;=\;
\sum_{\gamma\in\Gamma} e^{it c_H(\cdot,\gamma)} X_\gamma F_\gamma.
\]
But this is $\alpha_t^H(X)(\omega,\gamma)= e^{it c_H(\omega,\gamma)} X_\gamma(\omega)$, namely $\alpha_t^H=\alpha_t^{c_H},$
as claimed.
\end{proof}
As an immediate consequence of the previous Proposition, one gets:
\begin{corollary}\label{coro: KMS measure}
   The measure $\mu_\beta$, defined in \eqref{eq: KMS measure}, is a $(c_H,\beta)$-KMS measure. Consequently, the induced state $\SM_{\mu_\beta}\equiv \SM_\beta$ is a KMS state for the Kitaev dynamics at inverse temperature $\beta\in[0,\infty)$.
\end{corollary}

We are now in a position to state the main result of this work, which proves the existence and uniqueness of KMS states for all inverse temperatures $\beta$
\begin{theorem}\label{teo: unique kms}
The induced state $\SM_\beta$  by $\mu_\beta$ is the unique KMS state for the abelian Kitaev model at finite temperature $\beta\in [0,\infty).$
\end{theorem}
\begin{proof}
Let $\SM$ be any $\beta$-KMS state for the Kitaev Hamiltonian. The case $\beta=0$ is immediate, since the $0$-KMS condition coincides with the tracial condition and $\A$ has a unique tracial state. We therefore assume $\beta\neq 0$ in the following. By Theorem~\ref{theo: KMS measures}, there exists a
$(\alpha^{c_H},\beta)$-KMS measure $\mu$ on $\Omega$
such that $\SM=\SM_\mu$.
Hence, it suffices to prove that $\mu_\beta$ is the unique $(\alpha^{c_H},\beta)$-KMS measure. 

Any $\mu$ is completely determined by its values on the cylinder sets
\[
C(K,\omega')\;:=\;
\big\{\omega\in \Omega \mid \omega|_K=\omega'|_K\big\}, \quad K\in \Kk(W), \quad \omega' \in \Omega,
\]
and we have
\begin{equation}\label{Eq:Shift}
 \eta_{\gamma}^{-1}(C(K,\omega'))=C(K,\eta_{\gamma^{-1}}(\omega'))   
\end{equation} for all $\gamma \in \Gamma$. Fix enumerations
\[
V\;=\;\{v_1,v_2,\dots\},
\qquad
\widetilde V\;=\;\{\tilde v_1,\tilde v_2,\dots\},
\]
and define $
V_n=\{v_1,\dots,v_n\},$ $
\widetilde V_n=\{\tilde v_1,\dots,\tilde v_n\},$ with
$W_n=V_n\sqcup \widetilde V_n.$ Then $(W_n)_{n\ge1}$ is an increasing filtration of $W$ and $W=\bigcup_{n\ge1}W_n.$ 

For $(\chi,g)\in \widetilde G \times G$, set $\omega_{\chi,g} \in \Omega$ to be the configuration with values
\((\chi,g)\) at $(v_1,\tilde{v}_1)$, and values $1_{\widetilde{G}}$ and $1_G$ at all other coordinates $v_i$ and $\tilde v_i$. Also, 
set
$C(n,\chi,g)=C(W_n,\omega_{\chi,g})$.

Now, consider an arbitrary $n$-cylinder $C(W_n,\omega')$ and set $\prod_{i=1}^n \omega'(v_i)=\chi$ and $\prod_{i=1}^n \omega'(\tilde{v}_i)=g$. Let $\gamma\in\Gamma$ with  support contained in $W_n$ defined as $\gamma(v_1)=\chi\omega'(v_1)^{-1}$, $\gamma(\tilde{v}_1)=g\omega'(\tilde{v}_1)^{-1}$ and $\gamma(w)=\omega'(w)^{-1}$  for all the other $w\in W_n.$ Clearly, $\gamma$ is a well-defined element of $\Gamma$ since $\prod_{i=1}^n \gamma(v_i)=1_{\widetilde{G}}$ and $\prod_{i=1}^n \gamma(\tilde{v}_i)\;=\;1_G$. Note that $\gamma$ is entirely determined by $W_n$ and $\omega'$. It follows from \eqref{Eq:Shift}  that $C(W_n,\omega')=\eta_\gamma^{-1}\big(C(n,\chi,g)\big)$. Furthermore, from \eqref{eq: cocycle}, $c_H(\omega_1,\gamma)=c_H(\omega_2,\gamma)$ if $\omega_i$ coincide on the support of $\gamma$, and therefore $c_H(\omega,\gamma)=c_H(\omega_{\chi,g},\gamma)$ for any $\omega \in C(n,g,\chi)$. Hence, the KMS condition implies
\begin{equation}\label{eq:KMSTrick}
\begin{aligned}
\mu\big(C(W_n,\omega')\big)
& \;=\;\int\nolimits_{C(n,g,\chi)}e^{-\beta c_H(\omega,\gamma)}d\mu(\omega) \\
& \;=\;
e^{-\beta c_H(\omega_{\chi,g},\gamma)}\,\mu\big(C(n,g,\chi)\big).
\end{aligned}
\end{equation}
Therefore, at level n, any $(c_H,\beta)$-KMS measure $\mu$
is completely determined by the $|G|^2$ numbers $
\mu\big(C(n,\chi,g)\big)$.

 Define the vector
$u_n
:=
\bigl(\mu\big(C(n,\chi,g)\big)\bigr)_{g,\chi}
\in \mathbb R^{|G|^2}.$ Observe that
\[
C(n,\chi,g)
\;=\;
\bigcup_{\chi',g'}
C(n+1,\chi,g;\chi',g'),
\]
where $
C(n+1,\chi,g;\chi',g')=C(W_{n+1}, \omega_{\chi,g}^{\chi',g'})$, with  $\omega_{\chi,g}^{\chi',g'}\in \Omega$ denoting the configuration that takes the values \((\chi,g)\) at $(v_1,\tilde{v}_1)$, \((\chi',g')\) at $(v_{n+1},\tilde{v}_{n+1})$, and the neutral values $1_{\widetilde{G}}$ and $1_G$ at all other coordinates $v_i$ and $\tilde v_i$. These cylinders are disjoint, so one obtains
\[
\mu\big(C(n,g,\chi)\big)
\;=\;
\sum_{g',\chi'}
\mu\big(C(n+1,g,\chi;g',\chi')\big).
\]
Using \eqref{eq:KMSTrick}, one verifies that
\[
\mu\big(C(n+1,\chi,g;\chi',g')\big)
\;=\;
A_\beta(\chi,g;\zeta,h)\,
\mu\big(C(n+1,\zeta,h)\big),
\]
where $h=gg'$ and $\zeta=\chi\chi'$,
and 
\[
A_\beta(\chi,g;\zeta,h)
\;=\;
e^{-\beta\big(1+\delta_{1_G}(h)-\delta_{1_G}(g)-\delta_{1_G}(g^{-1}h)\big)}
\,
e^{-\beta\big(1+\delta_{1_{\widetilde G}}(\zeta)-\delta_{1_{\widetilde G}}(\chi)-\delta_{1_{\widetilde G}}(\chi^{-1}\zeta)\big)}
\]
Thus, $u_n=A_\beta\,u_{n+1},$
where the matrix $A_\beta=\big(A_\beta(\chi,g;\zeta,h)\big)_{\chi,g;\zeta,h}\in M_{|G|^2}(\mathbb R)$
is independent of $n$. As a consequence, $u_1=A_\beta^nu_n$ and 
\begin{equation}\label{eq: interesection}
   u_1\in \bigcap_{n\geq 1}A_\beta^n(\R^{|G|^2}_+) 
\end{equation}
Since $A_\beta$ has strictly positive entries, the Perron-Frobenius theorem implies that there exists a unique strictly positive eigenvector $u$ of $A_\beta$. By \eqref{eq: interesection} and the Perron--Frobenius convergence theorem, $
u_1 = \lambda u$ for some $\lambda > 0.$ Moreover, $\mu$ is a probability measure and the cylinders at level $1$
form a partition of $\Omega$, hence $
\sum_i (u_1)_{(i)} = 1.$
This normalization uniquely determines $\lambda$.
Consequently, $u_1$ is uniquely determined. Moreover,  the sequence $(u_n)_n$ is also uniquely determined since $A_\beta$ is invertible by Proposition \ref{prop: invertibility} and $u_n=A_{\beta}^{-n}u_1.$ Thus,  $\mu$ is unique since all its information is encoded in the sequence $(u_n)_n$. This completes the proof and verifies that $\mu=\mu_\beta.$
\end{proof}

\begin{remark}
  It follows from the proof of Theorem \ref{teo: unique kms} that the measure $\nu_\beta\otimes \tilde{\nu}_\beta$ on $G\times \widetilde{G}$ coincides, up to normalization, with the unique Perron-Frobenius eigenvector of the matrix $A_\beta.$ $\hfill \blacktriangleleft $
\end{remark}

It is possible to compute explicitly the expectation value of the  projections $P^\chi_v$ and $P^g_{\tilde v}$ under the KMS state $\SM_\beta$. Indeed, one obtains
$$\SM_\beta(P_v^\chi)\;=\;\frac{e^\beta}{e^\beta+(|G|-1)}\;=\;\SM_\beta(P_{\tilde v}^g),\qquad g= 1_G,\;\chi=1_{\widetilde{G}}$$
$$\SM_\beta(P_v^\chi)\;=\;\frac{1}{e^\beta+(|G|-1)}\;=\;\SM_\beta(P_{\tilde v}^g),\qquad g\neq 1_G,\;\chi\neq 1_{\widetilde{G}}$$
\begin{proposition}
The zero-temperature limit of the family of KMS states $(\SM_\beta)_{\beta\in [0,\infty)}$ converges to the unique frustration-free ground state $\SM_{\omega_0}$ of the Kitaev model. That is,
\[\text{weak}^
*\!-\!\lim_{\beta\to\infty}\SM_\beta
\;=\;
\SM_{\omega_0}.
\]
\end{proposition}
\begin{proof}
This is a direct consequence of that when $\beta\to\infty$, one has
\[
\frac{e^\beta}{e^\beta+(|G|-1)} \longrightarrow 1,
\qquad
\frac{1}{e^\beta+(|G|-1)} \longrightarrow 0.
\]
Hence $
\nu_\beta \longrightarrow \delta_{1_{\widehat G}},$
 and $\tilde\nu_\beta \longrightarrow \delta_{1_G},$ in the weak topology of probability measures. 
\end{proof}

We finish this section by verifying that the matrix $A_\beta$ is invertible.
\begin{proposition}\label{prop: invertibility}
    The matrix $A_\beta$ is invertible for any $\beta\neq 0.$
\end{proposition}
\begin{proof}
    Observe that $A_\beta(g,\chi;h,\zeta)
=B(g,h)\,D(\chi,\zeta)$, where
\[
B(g,h)
\;=\;
e^{-\beta\bigl(1+\delta_{1_G}(h)-\delta_{1_G}(g)-\delta_{1_G}(g^{-1}h)\bigr)},\;
D(\chi,\zeta)
\;=\;
e^{-\beta\bigl(1+\delta_{1_{\widetilde G}}(\zeta)-\delta_{1_{\widetilde G}}(\chi)-\delta_{1_{\widetilde G}}(\chi^{-1}\zeta)\bigr)}
\]
Hence $
A_\beta=B\otimes D$.   Set $q=e^{-\beta}$. A direct computation shows that\[
B(1_G,k)=1, \quad \forall k\in G,
\]
and for $g\neq 1_G$,
\[
B(g,1_G)=q^2,\qquad B(g,g)=1,\qquad
B(g,k)=q \quad (k\neq 1_G,g).
\]
Ordering the elements of $G$ with the identity first, the matrix $B$ takes the form
\[
B\;=\;
\begin{pmatrix}
1 & 1 & \cdots & 1\\
q^2 & 1 & \cdots & q\\
\vdots & \vdots & \ddots & \vdots\\
q^2 & q & \cdots & 1
\end{pmatrix}.
\]
Let $N:=|G|$. The matrix $B$ has the block form
\[
B\;=\;
\begin{pmatrix}
1 & \mathbf{1}^\top\\
q^2\mathbf{1} & C
\end{pmatrix},
\]
where $\mathbf{1}\in\mathbb{R}^{N-1}$ denotes the column vector of ones and
\[
C\;=\;
\begin{pmatrix}
1 & q & \cdots & q\\
q & 1 & \cdots & q\\
\vdots & \vdots & \ddots & \vdots\\
q & q & \cdots & 1
\end{pmatrix}
\;=\;(1-q)I_{N-1}+qJ_{N-1},
\]
with $J_{N-1}$ the $(N-1)\times (N-1)$ matrix of ones. Since $J_{N-1}\mathbf{1}=(N-1)\mathbf{1}$ and $J_{N-1}x=0$ for every $x\perp \mathbf{1}$, the eigenvalues of $C$ are
$1+(N-2)q$ and $1-q$ with multiplicity $N-2$. Hence $
\det(C)\;=\;(1-q)^{N-2}\bigl(1+(N-2)q\bigr),$
and moreover for $\beta\neq 0$
\[
C^{-1}\mathbf{1}\;=\;\frac{1}{1+(N-2)q}\mathbf{1}.
\]
Applying the Schur complement formula, we obtain
\[
\det(B)
\;=\;
\det(C)\,\det\!\bigl(1-\mathbf{1}^\top C^{-1}q^2\mathbf{1}\bigr).
\]
Since
\[
\mathbf{1}^\top C^{-1}\mathbf{1}
\;=\;
\frac{N-1}{1+(N-2)q},
\]
it follows that
\[
\det(B)
\;=\;
\det(C)\left(1-\frac{(N-1)q^2}{1+(N-2)q}\right).
\]
Substituting the value of $\det(C)$ and simplifying, we obtain
\[
\det(B)
\;=\;
(1-q)^{N-1}\bigl(1+(N-1)q\bigr)\;=\;(1-q)^{|G|-1}\bigl(1+(|G|-1)q\bigr).
\]
In particular, $B$ is invertible whenever $\beta\neq 0$. The same argument applies to $D$, yielding
\[
\det(D)\;=\;(1-q)^{|\widetilde G|-1}\bigl(1+(|\widetilde G|-1)q\bigr),
\]
so $D$ is also invertible for $\beta\neq 0$. Because $A_\beta=B\otimes D$, then $A_\beta$ is invertible for all $\beta\neq 0.$
\end{proof}

\section*{Declarations}
The authors have no conflicting or competing interests to declare that are relevant to the content of this article. No data was produced or used for this work.

\end{document}